\documentclass{comjnl}

\usepackage{}
\usepackage{amsfonts}
\usepackage{amsmath}
\usepackage{amssymb}
\usepackage{latexsym}
\usepackage{graphicx}

\usepackage[usenames]{color}

\input amssym.def
\newsymbol\rtimes 226F
\newfont{\nset}{msbm10}

\begin{document}

\title{Extended corona product as an exactly tractable model for weighted heterogeneous networks}

\author{Yi Qi}
\author{Huan Li}
\author{Zhongzhi~Zhang}\email{zhangzz@fudan.edu.cn}
\affiliation{Shanghai Key Laboratory of Intelligent Information
Processing, School of Computer Science, Fudan University, Shanghai 200433, China}
\shortauthors{Y. Qi, H. Li and Z. Zhang}


\keywords{Graph product, Corona product, Weighted complex network, Random walk, Graph spectra, Weighted spanning trees}

\begin{abstract}
Various graph products and operations have been widely used to construct complex networks with common properties of real-life systems. However, current works mainly focus on designing models of binary networks, in spite of the fact that many real networks can be better mimicked by heterogeneous weighted networks. In this paper, we develop a corona product of two weighted graphs, based on which and an observed updating mechanism of edge weight in real networks, we propose a minimal generative model for inhomogeneous weighted networks. We derive analytically relevant properties of the weighted network model, including strength, weight and degree distributions, clustering coefficient, degree correlations and diameter. These properties are in good agreement with those observed in diverse real-world weighted networks. We then determine all the eigenvalues and their corresponding multiplicities of the transition probability matrix for random walks on the weighted networks. Finally, we apply the obtained spectra to derive explicit expressions for mean hitting time of random walks and weighted counting of spanning trees on the weighted networks. Our model is an exactly solvable one, allowing to analytically treat its structural and dynamical properties, which is thus a good test-bed and an ideal substrate network for studying different dynamical processes, in order to explore the impacts of heterogeneous weight distribution on these processes.
\end{abstract}

\maketitle

\section{Introduction}
The last two decades have witnessed a mass of activity devoted to characterizing and understanding the structure of real-life networks~\cite{Ne03}. Extensive empirical studies have identified some universal properties shared by a variety of real systems, such as small-world effect~\cite{WaSt98} and scale-free behavior~\cite{BaAl99}.  Small-world effect is characterized by small average path length and large clustering coefficient~\cite{WaSt98}, while scale-free behavior means that the degree of nodes is heterogeneous, following a heavy-tail or power-law distribution~\cite{BaAl99}. In addition to these two topological aspects, many studies have also shown that a wealth of real networks synchronously exhibit a large heterogeneity in the distributions of both node strength and edge weight~\cite{BaBaPaVe04}, for example, scientific collaboration network~\cite{Ne01}, worldwide airport network~\cite{LiCa04,GuMoTuAm05}, and metabolic network~\cite{AlKoViOlBa04}. These striking structural and weighted properties play a crucial role in diverse dynamical processes taking place on networks~\cite{WaSt98,AlJeBa00, SaSaPa08,ChWaWaLeFa08,LiZh03,YiZhLiCh15}.


In parallel with the discoveries of common properties for real networks,  considerable attention has been paid to find generating mechanisms and models for networks that display the prominent features of real systems~\cite{PrWa11,ZhCo11,BaBaVe04,BaBaVe04PRE}. Since massive networks often consist of small pieces, for example, communities~\cite{GiNe02} and motifs~\cite{MiShItKaChAl02}, graph products and operations are a natural way to generate networks, by using which one can built a large network out of two or more smaller ones. In this perspective, many graph products have been employed in the design of realistic models, in order to generate real networks and capture their common properties, including Cartesian product~\cite{ImKl00}, hierarchical product~\cite{BaCoDaFi09,BaDaFiMi09,BaCoDaFi16}, corona product~\cite{LvYiZh15,ShAdMi17}, 
Kronecker product~\cite{We62,LeFa07,MaXu07, LeChKlFaGh10}, among others~\cite{PaNgBoKnFaRo11}. In addition, diverse graph operations were exploited to model complex networks~\cite{DoGoMe02,RaBa03,AnHeAnDa05,DoMa05}. However, most of current works focus on models for building unweighted networks, failing to match the properties of heterogeneous distributions of node strength and edge weight.

In this paper, we define an extended corona product for weighted graphs. Applying this generalized corona product and the reinforcement mechanism of edge weight in realistic networks, e.g. airport networks~\cite{LiCa04,GuMoTuAm05}, we introduce a simple generative model for heterogeneous weighted networks, which leads to rich topological and weighted properties. We offer an exhaustive analysis of the considered model and determine exactly its relevant properties, including strength, weight and degree distributions, clustering coefficient, degree correlations and diameter, which match the statistical properties shared by many realistic networks. We also characterize all the eigenvalues and their corresponding multiplicities of the transition probability matrix for random walks on the proposed weighted networks. Based on the obtained spectra, we further deduce closed-form expressions for average hitting time of biased random walks, as well as weighted counting of spanning trees on the networks, with the latter being consistent with the result derived by a different technique.

 Note that the standard  corona product has been previously applied to generate complex networks~\cite{LvYiZh15,ShAdMi17}. However, the resulting networks are binary,  and their degree follows an exponential form distribution that is almost homogenous. Moreover, for these networks, only the spectra for adjacency matrix and Laplacian matrix can be derived. In contrast, the proposed graphs are weighted, which are created by an extended corona product. Particularly, our graphs obey  heterogeneous distributions for vertex degree and strength, as well as the edge weight, as observed in many real networks. Another different aspect  for our weighted networks is that  the eigenvalues for transition probability matrix can be determined, instead of adjacency matrix and Laplacian matrix. Finally,  our networks are also largely different from those fractal  binary networks that have received considerable attention~\cite{GaSoMa07,GaSoMa08}.


\section{Construction of weighted heterogeneous networks}

Let $\mathcal{G}(\mathcal{V},\mathcal{E},w)$ be a simple connected weighted graph (network),  where $\mathcal{V}$ and $\mathcal{E}$ are sets of vertices (nodes) and edges, and $w: \mathcal{E} \to \mathbb{R}$ is a weight function. Let $N=|\mathcal{V}|$ and $M=|\mathcal{E}|$ denote, respectively, the number of vertices and edges in $\mathcal{G}(\mathcal{V},\mathcal{E},w)$, where the weight of an edge adjacent to vertices $i$ and $j$ is denoted by $w_{ij}$. Then, the strength $s_i$ of vertex $i$ in $\mathcal{G}(\mathcal{V},\mathcal{E},w)$ is defined as $s_i=\sum_j{w_{ij}}$~\cite{BaBaPaVe04}.

For unweighted (binary) simple graphs, Frucht and Harary proposed~\cite{FrHa70} the corona product of two graphs. Let $\mathcal{G}_1$ (with $z_1$ vertices) and $\mathcal{G}_2$ be two simple binary graphs. The corona $\mathcal{G}_1 \circ \mathcal{G}_2$ of $\mathcal{G}_1$ and $\mathcal{G}_2$ is a graph obtained by taking one copy of graph $\mathcal{G}_1$, $z_1$ copies of graph $\mathcal{G}_2$, and connecting the $i$th vertex of $\mathcal{G}_1$ and each vertex of the $i$th copy of $\mathcal{G}_2$, where $i=1,2,\ldots,z_1$.  This graph operation allows one to generate complex graphs from simple ones. The combinatorial and spectral properties  of corona product of two graphs have been much studied~\cite{FaHaVo08,BaPaSa07, McMc11}.

In a recent work~\cite{LaJaKi16}, a generalized corona of simple graphs was proposed. Given simple unweighted graphs $\mathcal{P}$ (with $z$ nodes) and $\mathcal{Q}_i$ ($i=1,2,\ldots,z$), the generalized corona of $\mathcal{P}$ and $\mathcal{Q}_i$ is the graph obtained by taking one replica of $\mathcal{P}$ and $\mathcal{Q}_i$ and joining every vertex of $\mathcal{Q}_i$ to the $i$th vertex of $\mathcal{P}$. In fact, this generalized corona is also applicable when $\mathcal{P}$ (with $z$ nodes) and $\mathcal{Q}_i$ ($i=1,2,\ldots,z$) are weighted graphs. Here we use this graph operation to construct heterogeneous weighted networks. To this end, we extend this corona product of unweighted graphs to some weighted graphs.
\begin{definition}\label{weightdef11}
Let $\mathcal{G}_1$ and $\mathcal{G}_2$ be two weighted graphs, with the strength $s_i$ of each vertex $i$ in $\mathcal{G}_1$ being an even number. Then the extended corona product of two weighted graphs, denoted by $\mathcal{G}_1 \circledcirc \mathcal{G}_2$, is a weighted graph constructed in the following way. For each vertex $i$ in $\mathcal{G}_1$ with strength $s_i$, take $\frac{s_i}{2}$ copies of $\mathcal{G}_2$, and link all vertices in each of $\frac{s_i}{2}$ replicas of $\mathcal{G}_2$ to $i$ by edges with unit weight. \end{definition}

Using the above defined extended corona product, coupled with the reinforce mechanism of edge weight, we can built an iteratively growing inhomogeneous  weighted  networks, with its topological and weighted properties matching those of realistic systems.

\begin{definition}\label{weightdef12}
Let $\mathcal{K}_2$ be a weighted graph with two vertices connected by one edge with unit weight. Then the iteratively growing heterogeneous  weighted  networks $\mathcal{W}_n$, $n \ge 0$, is constructed as follows. For $n=0$, $\mathcal{W}_0$ is a triangle consisting of three edges with unit weight. For $n \ge 1$, $\mathcal{W}_n$ is obtained from $\mathcal{W}_{n-1}$ by performing the following two operations.  \\
(I)  Generate a weight network $\mathcal{W}_{n-1}\circledcirc \mathcal{K}_2$ by applying the extended corona product of $\mathcal{W}_{n-1}$ and $\mathcal{K}_2$.\\
(II) For each old edge with weight $w$ in $\mathcal{W}_{n-1}\circledcirc \mathcal{K}_2$, that is, an edge belonging to $\mathcal{W}_{n-1}$, increase its weight by $\delta w$ ($\delta$ is non-negative integer), leading to $\mathcal{W}_{n}$.
\end{definition}

Note that the two operations in Definition~\ref{weightdef12} serve, respectively, as the strength driven attachment and weight reinterment (updating) mechanisms in real networks and the famous stochastic model for heterogeneous weighted networks~\cite{BaBaVe04,BaBaVe04PRE}. It is thus expected that our model exhibits similar properties as those of realistic networks and its random counterparts~\cite{BaBaVe04,BaBaVe04PRE}.


By Definition~\ref{weightdef12}, it is easy to verify that for $n \ge 1$, graph $\mathcal{W}_n$ can also be built from $\mathcal{W}_{n-1}$ in an iterative way as follows. First, for each existing triangle in $\mathcal{W}_{n-1}$ with weight $w$ for every edge, perform the following operations for each of its three vertices. We create $w$ groups of new vertices, with each group containing two vertices. Both vertices of each  group and their `mother'  vertex form a new triangle, each edge of which has an unit weight. Then, for each edge in $\mathcal{W}_{n-1}$, we increase its weight by $\delta$ times. The proof of the equivalence between this iterative construction and Definition~\ref{weightdef12} is straightforward, we thus omit the proof detail.   Figure~\ref{construct} illustrates the construction of network  $\mathcal{W}_n$. 

\begin{figure}
\centering
\includegraphics[width=0.60\linewidth,trim=0 0 0 0]{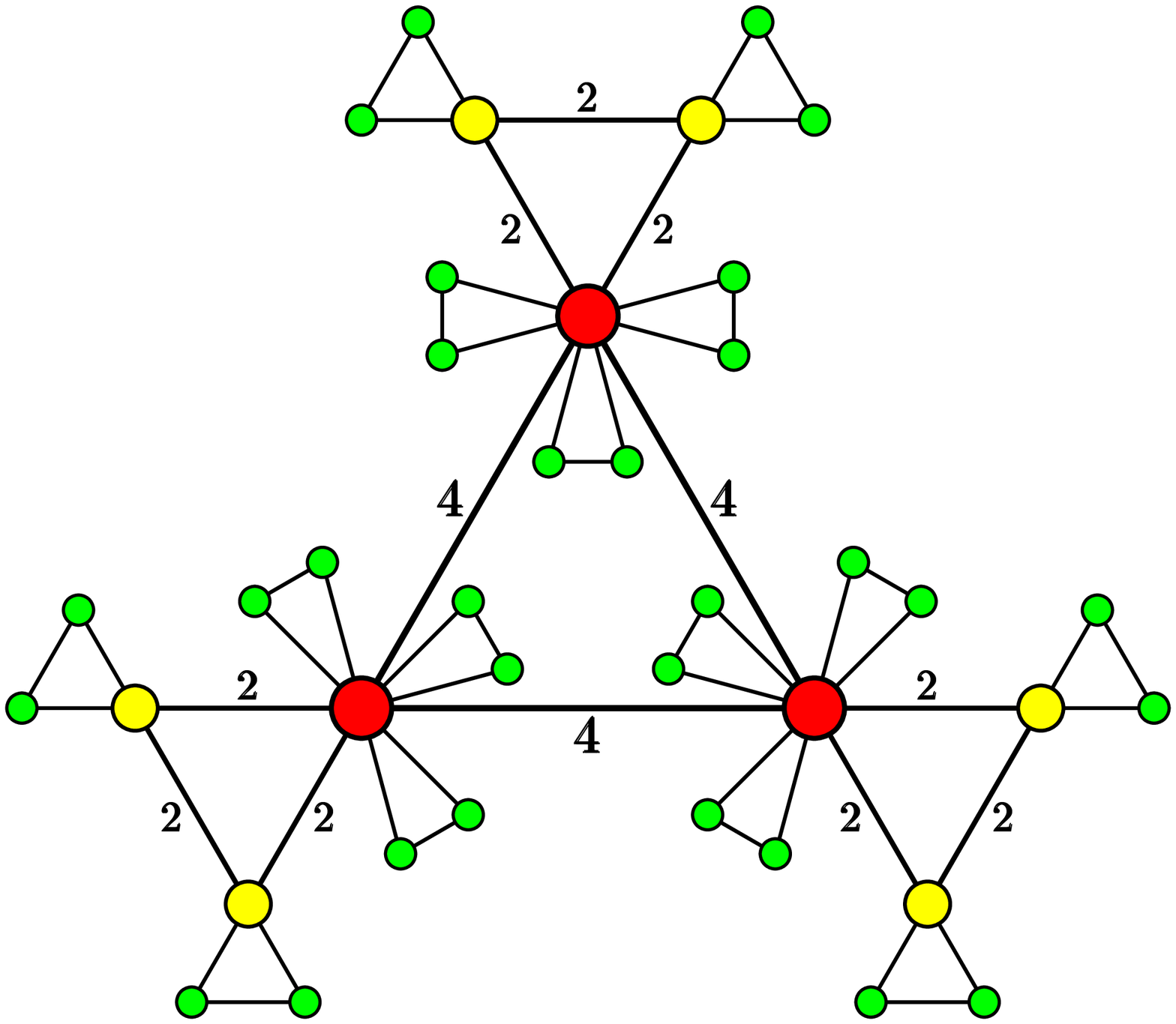}
\caption{Illustration of the graph $\mathcal{W}_2$ for $\delta = 1$. The bare edges denote those edges of unit weight.}\label{construct}
\end{figure}


Notice that when $\delta=0$, $\mathcal{W}_{n}$  is exactly the binary  scale-free small-world Koch network~\cite{ZhZhXiChLiGu09}, the properties of which have been extensively studied. Thus, in what follows, we only consider the case $\delta>0$.

The second construction method of the weighted networks allows to analytically treat their properties.
\begin{proposition}\label{prop:order}
In the graphs $\mathcal{W}_n$, the total number of vertices $N_n$, the total number of edges $E_n$, the total number of triangles $L_{\triangle}(n)$, and the total weight of all edges $W_n$,  are
\begin{equation}\label{wconstr1}
N_n=\frac{6(\delta+4)^n+3\delta+3}{\delta+3},
\end{equation}
\begin{equation}\label{wconstr2}
E_n=\frac{9(\delta+4)^n+3\delta}{\delta+3},
\end{equation}
\begin{equation}\label{wconstr10}
L_{\triangle}(n)=\frac{3(\delta+4)^n+\delta}{\delta+3}
\end{equation}
and
\begin{equation}\label{wconstr3}
W_n=3(\delta+4)^n,
\end{equation}
respectively.
\end{proposition}
\begin{proof}
Let $n_v(n)$, $n_e(n)$, and $l_{\triangle}(n)$ denote, respectively, the number of vertices, edges, and triangles generated at $n$th iteration. Note that the addition of every vertex group leads to $2$ new vertices $3$ new edges, so the relation $n_e(n) = \frac{3}{2}n_v(n)$ holds. By construction, for $n>0$, we have
\begin{equation}\label{wconstr4}
n_v(n)=2W_{n-1},
\end{equation}
\begin{equation}\label{wconstr5}
E_n=E_{n-1}+\frac{3}{2}n_v(n),
\end{equation}
\begin{equation}\label{wconstr11}
l_{\triangle}(n)=W_{n-1},
\end{equation}
and
\begin{equation}\label{wconstr6}
W_n=(1+\delta)W_{n-1}+3W_{n-1}.
\end{equation}
On the right-hand side (rhs) of Eq.~\eqref{wconstr6}, the first term accounts for the sum of weight of the old edges, while the second term represents the total weight of the new edges generated at iteration $n$. Considering the initial condition $W_0 = 3$, Eq.~\eqref{wconstr6} is solved to yield
\begin{equation}\label{wconstr7}
W_n=3(\delta+4)^n.
\end{equation}
Substituting Eq.~\eqref{wconstr7} into Eqs.~\eqref{wconstr4} and~\eqref{wconstr11} and considering the relation $n_e(n) = \frac{3}{2}n_v(n)$ give
\begin{equation}
n_v(n)=6(\delta+4)^{n-1}, \nonumber
\end{equation}
\begin{equation}
n_e(n)=9(\delta+4)^{n-1}, \nonumber
\end{equation}
and
\begin{equation}
l_{\triangle}(n)=3(\delta+4)^{n-1}. \nonumber
\end{equation}
Then in network $\mathcal{W}_n$ the total number of vertices is
\begin{equation}\label{wconstr8}
N_n=\sum_{n_i=0}^{n}n_v(n_i)=\frac{6(\delta+4)^n+3\delta+3}{\delta+3}, \nonumber
\end{equation}
and the total number of triangles is
\begin{equation}\label{wconstr12}
L_{\triangle}(n)=\sum_{n_i=0}^{n}l_{\triangle}(n_i)=\frac{3(\delta+4)^n+\delta}{\delta+3}.\nonumber
\end{equation}
Combining Eqs.~\eqref{wconstr4},~\eqref{wconstr5} and~\eqref{wconstr7} and considering the initial condition $E_0=3$, we obtain
\begin{equation}\label{wconstr9}
E_n=\frac{9(\delta+4)^n+3\delta}{\delta+3}.\nonumber
\end{equation}
The proof is completed.
\end{proof}

Thus, the average degree in network $\mathcal{W}_n$ is $\frac{2E_n}{N_n}$, which
is approximately equal to $3$ for large $n$.

\section{Structural and weighted properties}

In this section, we study the topological and weighted characteristics of the weighted networks $\mathcal{W}_n$.

\subsection{Strength distribution}


The strength distribution $P(s)$ of a weighted graph is the probability
that a randomly chosen node has strength $s$. When a network has a discrete sequence of vertex strength, one can also use cumulative strength distribution $P_{\rm cum}(s)$ instead of strength distribution~\cite{Ne03}, which is the probability that a vertex has strength greater than or equal to $s$, that is
\begin{equation}\label{dd0}
P_{\rm cum}(s)= \sum_{s'=s}^{\infty}P(s').\nonumber
\end{equation}
For a network with a power-law strength distribution $P(s)\sim s^{-\gamma}$, their  cumulative strength distribution is also power-law obeying $P_{\rm cum}(s)\sim s^{-(\gamma-1)}$.



\begin{proposition}\label{weightstrdis}
The strength distribution of the graphs $\mathcal{W}_n$ follows a power-law form $P_{\rm cum}(s)\sim s^{-\gamma_s}$ with the exponent  $\gamma_s=1+\frac{\ln(\delta+4)}{\ln(\delta+2)}$.
\end{proposition}

\begin{proof}
In $\mathcal{W}_n$, all simultaneously emerging nodes have identical strength. Let $s(n_i,n)$ denote the strength of a vertex $i$ in  $\mathcal{W}_n$, which was  generated at the $n_i$th iteration, then $s(n_i,n_i) = 2$. In order to determine $s(n_i,n)$, we introduce the quantity $\Delta s(n_i,n)$ to represent the difference between $s(n_i,n)$ and $s(n_i,n-1)$.
By construction,
\begin{align}\label{weightstrdis2}
\Delta s(n_i,n)&=s(n_i,n)-s(n_i,n-1)\nonumber \\
&=\delta s(n_i,n-1)+s(n_i,n-1)\nonumber \\
&=(\delta+1)s(n_i,n-1).
\end{align}
On the rhs of the second line of Eq.~\eqref{weightstrdis2}, the first item describes the increase of weight of the old edges connecting $i$ and those vertices already existing at iteration $n-1$, while the second term accounts for the total weight of the new edges incident to vertex $i$, each of which is generated at iteration $n$ and has unit weight.

Equation~\eqref{weightstrdis2} implies the following recursive relation:
\begin{equation}\label{weightstrdis3}
s(n_i,n)=(\delta+2)s(n_i,n-1).
\end{equation}
Using $s(n_i,n_i)=2$, we have
\begin{equation}\label{weightstrdis4}
s(n_i,n)=2(\delta+2)^{n-n_i}.
\end{equation}
Thus, the cumulative strength distribution
of $\mathcal{W}_n$ can be represented as~\cite{Ne03}
\begin{equation}\label{weightstrdis5}
P_{\rm cum}(s)=\sum_{\mu \le n_i}\frac{n_v(\mu)}{N_n}=\frac{6(\delta+4)^{n_i}+3\delta+3}{6(\delta+4)^n+3\delta+3}
\end{equation}
From Eq.~\eqref{weightstrdis4}, we can obtain
\begin{equation}\label{weightstrdis6}
n_i=n-\frac{\ln\frac{s(n_i,n)}{2}}{\ln(\delta+2)},\nonumber
\end{equation}
plugging which into the Eq.~\eqref{weightstrdis5} yields
\begin{equation}\label{weightstrdis7}
P_{\rm cum}(s)=\frac{6(\delta+4)^{n}\left(\frac{s}{2}\right)^{-\frac{\ln(\delta+4)}{\ln(\delta+2)}}+3\delta+3}{6(\delta+4)^n+3\delta+3}.\nonumber
\end{equation}
For large $n$, we have
\begin{equation}\label{weightstrdis8}
P_{\rm cum}(s)\sim \left(\frac{s}{2}\right)^{-\frac{\ln(\delta+4)}{\ln(\delta+2)}}.\nonumber
\end{equation}
Therefore, the strength of vertices in the graphs $\mathcal{W}_n$ obeys a power-law form with exponent $\gamma_s=1+\frac{\ln{(\delta+4)}}{\ln{(\delta+2)}}$.
\end{proof}

\subsection{Degree distribution}




In a similar way, we can obtain the  degree distribution of the weighted graphs $\mathcal{W}_n$.

\begin{proposition}\label{weightdegdis}
The  degree distribution of the graphs $\mathcal{W}_n$ exhibits a power law behavior $P_{\rm cum}(k)\sim k^{-\gamma_k}$ with $\gamma_k=1+\frac{\ln(\delta+4)}{\ln(\delta+2)}$.
\end{proposition}

\begin{proof}
In $\mathcal{W}_n$, the degree of all simultaneously emerging vertices is the
same. Let $k(n_i,n)$ be the degree of a vertex $i$ in $\mathcal{W}_n$, which was added to the graph at iteration $n_i$. By definition, $k(n_i,n_i) = 2$. According to network construcion, the degree $k(n_i,n)$ evolves as
\begin{equation}\label{wdegdis1}
k(n_i,n)=k(n_i,n-1)+s(n_i,n-1),\nonumber
\end{equation}
which, together with Eq.~\eqref{weightstrdis4}, leads to
\begin{equation}\label{wdegdis2}
k(n_i,n)=k(n_i,n_i)+\sum_{u=n_i+1}^{n} s(u,n)=\frac{2(\delta+2)^{n-n_i}+2\delta}{\delta+1}.\nonumber
\end{equation}
Then,  the cumulative degree distribution of $\mathcal{W}_n$  can be expressed as
\begin{align}\label{wdegdis3}
P_{\rm cum}(k)&=\frac{1}{N_n}\sum_{u\le n_i} n_v(u)\nonumber \\
&=\frac{6(\delta+4)^{n_i}+3\delta+3}{6(\delta+4)^n+3\delta+3}\nonumber \\
&=\frac{6(\delta+4)^{n}\left(\frac{\delta+1}{2}k-\delta\right)^{-\frac{\ln(\delta+4)}{\ln(\delta+2)}}+3\delta+3}{6(\delta+4)^n+3\delta+3}. \nonumber
\end{align}
For large $n$, we have
\begin{equation}\label{wdegdis4}
P_{\rm cum}(k)\sim \left(\frac{\delta+1}{2}k\right)^{-\frac{\ln(\delta+4)}{\ln(\delta+2)}}\,,\nonumber
\end{equation}
which means that the degree of  graph  $\mathcal{W}_n$ follows a power law distribution with the exponent identical to $\gamma_s$, i.e. $\gamma_k=\gamma_s=1+\frac{\ln(\delta+4)}{\ln(\delta+2)}$.
\end{proof}

\subsection{Weight distribution}


In addition to distributions of degree and strength, the weight distribution for the graphs $\mathcal{W}_n$ can also be analytically determined.
\begin{proposition}\label{wdegdis}
The weight of edges in the graphs $\mathcal{W}_n$  follows a power law distribution with exponent $\gamma_w=1+\frac{\ln{(\delta+4)}}{\ln{(\delta+1)}}$.
\end{proposition}

\begin{proof}
Let $w_e(n_i,n)$ be the weight of edge $e$ in $\mathcal{W}_n$,  which was generated at the iteration $n_i$, then $w_e(n_i,n_i)=1$. Since all the edges in  $\mathcal{W}_n$ emerging
simultaneously have the same weight, we can establish the recursive relation as follows.
\begin{equation}\label{wweigdis1}
w_e(n_i,n)=(1+\delta)w_e(n_i,n-1).
\end{equation}
Considering $w_e(n_i,n_i)=1$, Eq.~\eqref{wweigdis1} is solved to obtain
\begin{equation}\label{wweigdis2}
w_e(n_i,n)=(1+\delta)^{n-n_i}.
\end{equation}
Hence, the cumulative weight distribution of $\mathcal{W}_n$ is
\begin{equation}\label{wweigdis3}
P_{\rm cum}(w)=\sum_{\mu \le n_i}\frac{n_e(\mu)}{E_n}=\frac{9(\delta+4)^{n_i}+3\delta}{9(\delta+4)^n+3\delta}.
\end{equation}
From Eq.~\eqref{wweigdis2}, we can derive
\begin{equation}\label{wweigdis4}
n_i-n=\frac{\ln w_e(n_i,n)}{\ln (\delta+1)},
\end{equation}
Substituting which into Eq.~\eqref{wweigdis3} gives
\begin{equation}\label{wweigdis5}
P_{\rm cum}(w)=\frac{9(\delta+4)^{n}w^{-\frac{\ln(\delta+4)}{\ln(\delta+1)}}+3\delta}{9(\delta+4)^n+3\delta}.\nonumber
\end{equation}
Therefore, for large $n$, we have
\begin{equation}\label{wweigdis6}
P_{\rm cum}(w)\sim w^{-\frac{\ln(\delta+4)}{\ln(\delta+1)}},\nonumber
\end{equation}
which implies that the  weight distribution of $\mathcal{W}_n$  exhibits a power-law form $\gamma_w=1+\frac{\ln{(\delta+4)}}{\ln{(\delta+1)}}$.
\end{proof}

Note that in some previous random models for weighted networks~\cite{BaBaVe04,BaBaVe04PRE}, their distributions for vertex strength, vertex degree, and edge weight also display power-law forms. These heterogeneous distributions are consistent with those observed in realistic networks~\cite{BoLaMoChHw06}.

\subsection{Clustering coefficient and weighted clustering coefficient}

In a graph $\mathcal{G}$,  the clustering coefficient $C_v$ of a vertex $v$ with degree $k_v$  is defined~\cite{WaSt98} as the ratio between the number $\triangle_v$ of existing triangles including vertex $v$ and the total number of possible triangles including $v$, that is $C_v=\frac{2\triangle_v}{k_v(k_v-1)}$.
When $\mathcal{G}$ is a  weighted graph, the weighted clustering coefficient~\cite{BaBaPaVe04} of vertex $v$, denoted by $C_v^w$, is defined as
\begin{equation}\label{wcst10}
C_v^w=\frac{1}{s_v(k_v-1)}\sum_{j,h}\frac{w_{vj}+w_{vh}}{2}a_{vj}a_{vh}a_{jh},
\end{equation}
where $a_{xy}$ is the $xy$th entry of the adjacent matrix of graph $\mathcal{G}$ defined as follows: $a_{xy}=1$ if there exists an edge connecting vertex $x$ and vertex $y$, and $a_{xy}=0$ otherwise.

The clustering coefficient of the whole graph $\mathcal{G}$, denoted as $C(\mathcal{G})$, is defined as the average of $C_v$ over all vertices in the graph:
$C(\mathcal{G})=\frac{1}{N}\sum_{v \in \mathcal{V}}C_v$. When  $\mathcal{G}$ is a weighted graph, we can analogously define weighted clustering coefficient of $\mathcal{G}$.


Next we will calculate the clustering coefficient,  weighted clustering coefficient for every vertex and their average value in $\mathcal{W}_n$.

\begin{proposition}\label{wcluster}
For any vertex with degree $k$ in the graphs $\mathcal{W}_n$,  its clustering coefficient is $\frac{1}{k-1}$.
\end{proposition}
\begin{proof}
For an arbitrary  vertex $v$ in the graphs $\mathcal{W}_n$, the number of existing triangles $\triangle_v$ including $v$ and its degree $k_v$ satisfy relation $k_v=2\triangle_v$.
Thus, for any vertex in the graphs $\mathcal{W}_n$, there is a one-to-one correspondence between its clustering coefficient and its degree: For a vertex of degree $k$, its clustering coefficient is $\frac{1}{k-1}$.
\end{proof}

Hence, for a vertex with a large degree, its clustering coefficient is inversely proportional to its degree. Such a scaling has been observed in various real-life  networks~\cite{RaBa03}.

\begin{proposition}\label{wwghtcluster}
For any vertex with degree $k$ in the graphs $\mathcal{W}_n$,  its weighted clustering coefficient is $\frac{1}{k-1}$, independent of its strength.
\end{proposition}
\begin{proof}
For a vertex $i$ in the graphs $\mathcal{W}_n$ that was created at the $n_i$th iteration, its strength is $s(n_i,n)=2(\delta+2)^{n-n_i}$, its degree is $k(n_i,n)=\frac{2(\delta+2)^{n-n_i}+2\delta}{\delta+1}$,  and the number of triangles including $i$ is also $\frac{k(n_i,n)}{2}$. Furthermore, for each triangle, the weight of its three edges is the same. By construction, among all the $\frac{k(n_i,n)}{2}$ triangles attached to vertex $v$, the number of triangles with edge weight $1$, $1+\delta$, $\cdots$, $(1+\delta)^{n-n_i-1}$, $(1+\delta)^{n-n_i}$, equals, respectively, $(2+\delta)^{n-n_i-1}$, $(2+\delta)^{n-n_i-2}$, $\cdots$, $(2+\delta)^{0}$, $1$. Thus, the sum in Eq.~\eqref{wcst10} can be evaluated as
\begin{align}\label{wcst12}
\sum_{j,h}\frac{w_{ij}+w_{ih}}{2}a_{ij}a_{ih}a_{jh}&=2(1+\delta)^{n-n_i}+\sum_{u=0}^{n-n_i-1}2(\delta+2)^u(\delta+1)^{n-n_i-1-u}\nonumber\\
&=2(\delta+2)^{n-n_i},
\end{align}
which is equal to the strength of vertex $i$. Thus, for any vertex with degree $k$ in  graph $\mathcal{W}_n$, its weighted clustering coefficient is
$\frac{1}{k-1}$, which does not depend on the strength of the vertex.
\end{proof}

Propositions~\ref{wcluster} and~\ref{wwghtcluster} show that for any vertex in graph $\mathcal{W}_n$, its weighted clustering coefficient and its weighted clustering coefficient are equal to each other, signaling  that there exist no correlations between weights and topology with respect to the clustering coefficient of a single vertex. Moreover, both the clustering coefficient and  weighted clustering coefficient of the whole graph are also equal.

\begin{proposition}\label{wclustercoefficent}
The clustering coefficient $C(\mathcal{W}_n)$ of the graphs $\mathcal{W}_n$ is
\begin{equation}\label{wcst2}
C(\mathcal{W}_n)=\frac{\delta+3}{6(\delta+4)^n+3\delta+3}\left(\sum_{i=1}^{n}\frac{6(\delta+4)^i(\delta+1)}{2(\delta+2)^{n-i}+\delta-1}+\frac{3\delta+3}{2(\delta+2)^{n}+\delta-1}\right).
\end{equation}
\end{proposition}
\begin{proof}
As shown above, in $\mathcal{W}_n$ the degree sequence is discrete. The number of vertices with degree $2$, $4$, $\cdots$, $\frac{2(\delta+2)^{n-1}+2\delta}{\delta+1}$, $\frac{2(\delta+2)^{n}+2\delta}{\delta+1}$ is equal to $6(\delta+4)^{n-1}$, $6(\delta+4)^{n-2}$, $\cdots$, $6$, $3$, respectively. By Propositions~\ref{wcluster}, the clustering coefficient of any  vertex with degree $k$ is $\frac{1}{k-1}$. According to the definition of clustering coefficient of a graph, the proposition follows immediately.
\end{proof}

In Fig.~\ref{cluster}, we report $C(\mathcal{W}_n)$ as a function of $\delta$ and $n$, which shows that for large graphs, $C(\mathcal{W}_n)$ approaches to a high constant increasing with $\delta$.  For example, for $\delta=1,2, 3,$ and $4$,   $C(\mathcal{W}_\infty)$ tends to $0.8571, 0.8818, 0.8993$ and $0.9124$, respectively. Therefore, the whole family of graph $\mathcal{W}_n$ is highly clustered.


\begin{figure}
\centering
\includegraphics[width=0.85\linewidth,trim=0 0 0 0]{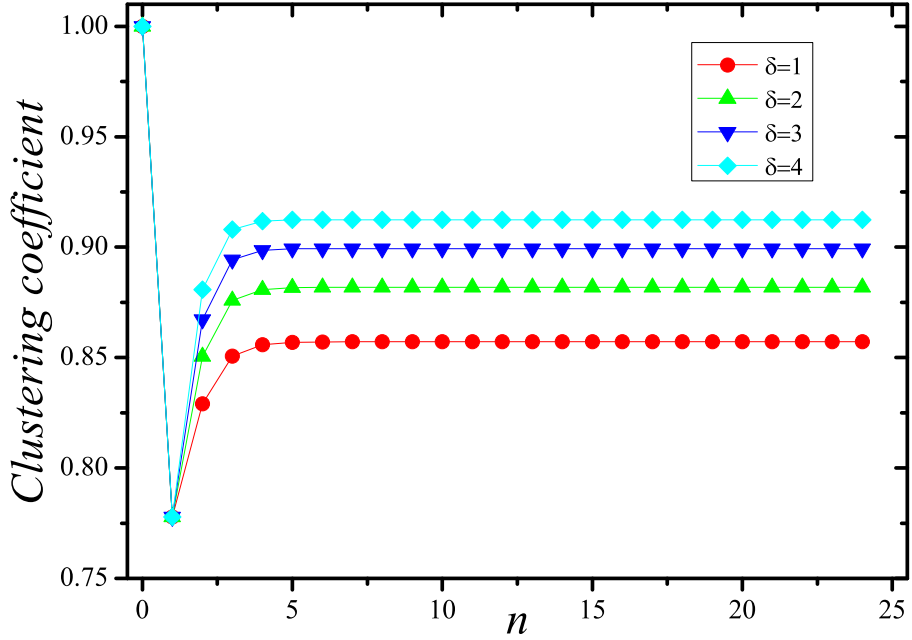}
\caption{Clustering coefficient of $\mathcal{W}_n$ for various $\delta$ and $n$.}\label{cluster}
\end{figure}

\subsection{Degree correlations}

Degree correlations are another important characteristic of a graph~\cite{SeSn02,Ne02}. In this subsection, we address the degree correlations of the proposed model for weighted graphs.

\subsubsection{Average nearest-neighbor degree}

A key quantity related to degree correlations~\cite{BrKaPo01} of a graph is the average degree of  nearest neighbors for vertices with degree $k$, denoted as $k_{\rm nn}(k)$. When $k_{\rm nn}(k)$ increases with $k$, it implies that vertices have a tendency to link to vertices with a similar or larger degree. In this situation, the graph is said to be assortative~\cite{Ne02}. In contrast, if $k_{\rm nn}(k)$ decreases with $k$, it means that vertices with  large degree have a high probability of being linked to those vertices with  small degree, and the graph is defined as disassortative. If correlations are absent, $k_{\rm nn}(k)$ is independent of $k$.

\begin{proposition}\label{weightdgco}
In the graphs $\mathcal{W}_n$, the average degree of nearest neighbors for vertices with degree $k$ is
\begin{small}
\begin{align}\label{weightdg1}
k_{\rm nn}(k) &= \frac{(\delta+1)(\delta+2)^{2n}(\delta+4)^{1-n}\left(\frac{(\delta+1)k-2\delta}{2} \right)^{\frac{\ln(\delta+4)}{\ln(\delta+2)}-1}-(\delta+2)((\delta+1)k-2\delta)}{\delta(\delta+3)(\delta+1)k/2}\nonumber\\
&\quad +\frac{2\delta}{(\delta+1)k}+\frac{\delta-1}{\delta+1}+\frac{2 (\delta(k-2)+k) \left(\frac{\ln \left((\delta  (k-2)+k)/2\right)}{\ln (\delta +2)}+\delta +2\right)}{(\delta +1) (\delta +2) k}.
\end{align}
\end{small}
\end{proposition}
\begin{proof}
By construction, for a vertex in $\mathcal{W}_n$, all its connections to vertices with  larger degree are made at the creation iteration when the vertex is generated, while the connections to vertices with smaller degree are made at each subsequent iteration. Then, for those vertices generated at the iteration $n_i \ge 1$ with degree $k=\frac{2(\delta+2)^{n-n_i}+2\delta}{\delta+1}$, $k_{\rm nn}(k)$ can be computed by
\begin{align}\label{weightdg2}
k_{\rm nn}(k)& =\frac{1}{n_v(n_i)k(n_i,n)}\Bigg(\sum_{u=0}^{n_i-1}n_v(u)s(u,n_i-1)k(u,n) \nonumber\\
&\quad +\sum_{u=n_i+1}^{n}n_v(n_i) s(n_i,u-1)k(u,n)\Bigg)+1,
\end{align}
where $k(n_i,n)$ denotes the degree of a vertex in $\mathcal{W}_n$, which was generated at iteration $n_i$. The first sum on the rhs of Eq.~\eqref{weightdg2} describes the links made to vertices with larger degree (i.e., $0\le u \le n_i-1$) when the vertices were generated at iteration $n_i$. The second term accounts for the links made to vertices with small degree at iteration $u$ ($n_i+1\le u \le n$). The last term $1$ explains the link connected to the simultaneously emerging vertex.

Substituting Eqs.~\eqref{weightstrdis4} and~\eqref{wdegdis2} into Eq.~\eqref{weightdg2},  we simplify Eq.~\eqref{weightdg2} to
\begin{align}\label{weightdg3}
k_{\rm nn}(k)&=\frac{(\delta+1)(\delta+2)^{n+n_i}(\delta+4)^{1-n_i}-2(\delta+2)^{1+n-n_i}+\delta^2(\delta+3)}{\delta(\delta+3)((\delta+2)^{n-n_i}+\delta)}\nonumber\\
&\quad+\frac{\delta-1}{\delta+1}+\frac{2(2+\delta+n-n_i)}{2+\delta+\delta(2+\delta)^{1+n_i-n}},
\end{align}
after some algebraic manipulations. Considering $k=\frac{2(\delta+2)^{n-n_i}+2\delta}{\delta+1}$, we can write $k_{\rm nn}(k)$  in terms of $k$ to obtain the result.
\end{proof}

Eq.~\eqref{weightdg1} shows  that in large graphs  $\mathcal{W}_n$ (i.e. $n \to \infty$), $k_{\rm nn}(k) \sim  k^{\frac{\ln(\delta+4)}{\ln(\delta+2)}-2}$. That is, $k_{\rm nn}(k)$ is approximately a power-law function of degree $k$ with negative exponent $\frac{\ln(\delta+4)}{\ln(\delta+2)}-2<0$ (since $\delta >0$), indicating that the graph family $\mathcal{W}_n$ is disassortative.

\subsubsection{Weighted average nearest-neighbor degree}

For a vertex $i$ with degree $k$ in a weighted network, its weighted average nearest-neighbor  degree is defined as~\cite{BaBaPaVe04}
\begin{eqnarray}\label{weightdg5}
k_{{\rm nn},i}^w(k)=\frac{1}{s_i}\sum_{j=1}^{N}w_{ij}k_j,\nonumber
\end{eqnarray}
while the global weighted degree correlations of the network can be defined as the average of weighted nearest-neighbor  degree $k_{\rm nn}^w(k)$ over all vertices with degree $k$, given by
\begin{eqnarray}\label{weightdg6}
k_{\rm nn}^w(k)=\langle k_{{\rm nn},i}^w(k)\rangle_k.\nonumber
\end{eqnarray}
The behavior of the metric $k_{\rm nn}^w(k)$ describes the weighted assortative or disassortative features considering the actual interactions among the elements of a system.

\begin{proposition}\label{weightwdgco}
In the graphs $\mathcal{W}_n$, the global weighted average degree of the nearest neighbors for vertices with degree $k$ is
\begin{align}\label{weightdg7}
k_{\rm nn}^w(k)&= \frac{\left(\frac{k(\delta+1)}{2}-\delta\right)^{\frac{\ln(\delta+1)}{\ln(\delta+2)}}}{\delta+1}
+\frac{2\left(\frac{\left(\frac{k(\delta+1)}{2}-\delta\right)^{\frac{\ln(\delta+1)}{\ln(\delta+2)}}}{\delta+1}+\delta^2+\delta-1\right)}{\delta(\delta+2)}-\nonumber\\
&\quad\frac{2(\delta+2)\left(\frac{k(\delta+1)}{2}-\delta\right)^{\frac{\ln(\delta+1)}{\ln(\delta+2)}}}{\delta(\delta+3)(\delta+1)}
+\frac{(\delta+2)^{2n}k^{\frac{\ln[(\delta+4)(\delta+1)]}{\ln(\delta+2)}-2}}{\delta(\delta+3)(\delta+4)^{n-1}}.
\end{align}
\end{proposition}
\begin{proof}
Analogously to computation of $k_{\rm nn}(k)$, for those vertices in $\mathcal{W}_n$ with degree $k=\frac{2(\delta+2)^{n-n_i}+2\delta}{\delta+1}$ that are generated at the iteration $n_i \ge 1$, $k_{\rm nn}(k)$, the global weighted average degree of their nearest neighbors can be calculated by
\begin{small}
\begin{align}\label{weightdg8}
k_{\rm nn}^w(k)&=\frac{1}{n_v(n_i)s(n_i,n)} \Bigg(\sum_{u=0}^{n_i-1}n_v(u)s(u,n_i-1)k(u,n)(1+\delta)^{n-n_i}+n_v(n_i)k(n_i,n)\nonumber\\
&\quad (1+\delta)^{n-n_i}+\sum_{u=n_i+1}^{n}n_v(n_i)s(n_i,u-1)k(u,n)(1+\delta)^{n-u}\Bigg).
\end{align}
\end{small}
After some algebraic manipulations, we obtain
\begin{align}\label{weightdg9}
k_{\rm nn}^w(k)&= \frac{2[(\delta+1)^{n-n_i-1}+\delta^2+\delta-1]}{\delta(\delta+2)}-\frac{2(\delta+2)(\delta+1)^{n-n_i-1}}{\delta(\delta+3)}+\nonumber\\
&\quad (\delta+1)^{n-n_i-1}+\frac{(\delta+1)^{n-n_i}(\delta+2)^{2n_i}(\delta+4)^{1-n_i}}{\delta(\delta+3)}.
\end{align}
Writing the above equation in terms of the vertex degree $k$, it is straightforward to get Eq.~\eqref{weightdg7}.
\end{proof}

According to Proposition~\ref{weightdgco}, we can see that for large network $\mathcal{W}_{n}$, the global weighted degree correlation $k_{\rm nn}^w(k)$ follows a power-law form $k_{\rm nn}^w(k)\sim k^{\frac{\ln [(\delta+4)(\delta+1)]}{\ln(\delta+2)}-2}$. Since $\delta>0$,  $\frac{\ln [(\delta+4)(\delta+1)]}{\ln(\delta+2)}-2>0$. Hence, different from the topological $k_{\rm nn}(k)$, the weighted $k_{\rm nn}^w(k)$ exhibits an assortative behavior in the whole  degree spectrum.

In Fig.~\ref{correlation}, we plot  $k_{\rm nn}(k)$ and $k_{\rm nn}^w(k)$ of the graphs $\mathcal{W}_{10}$ for different $\delta$. From Fig.~\ref{correlation}, we can see that the topological $k_{\rm nn}(k)$ is disassortative, while the weighted $k_{\rm nn}^w(k)$ is assortative. Moreover, for any given degree $k$ in $\mathcal{W}_{10}$, $k_{\rm nn}^w(k) \ge k_{\rm nn}(k)$, implying that edges with larger weights are pointing to neighbors with larger degree.

\begin{figure}
\centering
\includegraphics[width=0.85\linewidth,trim=0 0 0 0]{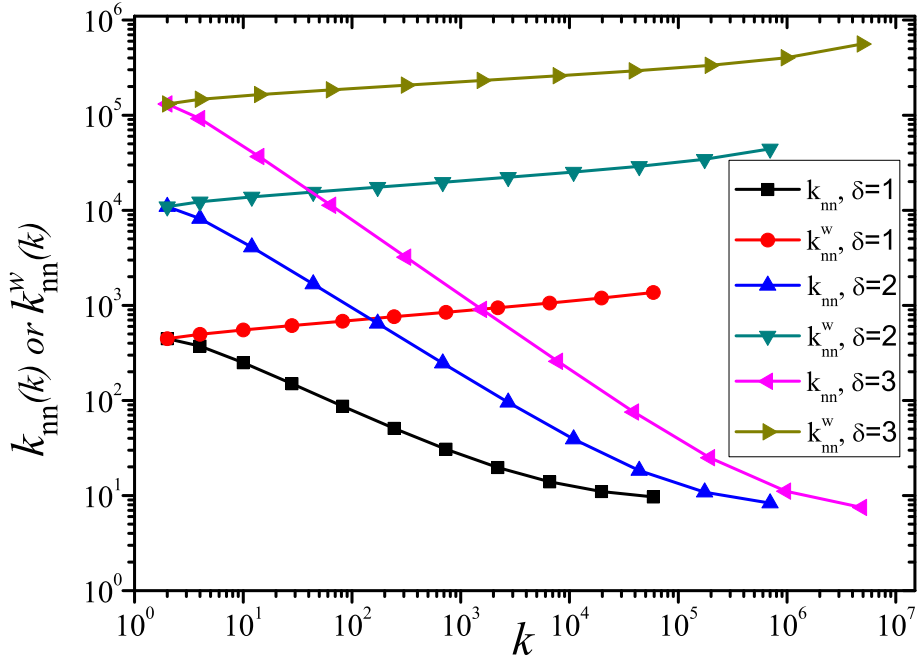}
\caption{Topological $k_{\rm nn}(k)$ and weighted $k_{\rm nn}^w(k)$  for the graphs $\mathcal{W}_{10}$ with various $\delta$.}\label{correlation}
\end{figure}

\subsection{Diameter}

The diameter of a graph is defined as the maximum of the shortest distances between all pairs of
vertices. 

\begin{proposition}\label{weightdia}
The diameter of the graphs $\mathcal{W}_n$ is $Diam(\mathcal{W}_n) = 2n+1$.
\end{proposition}
\begin{proof}
When $n=0$, $\mathcal{W}_0$ is a triangle,  implying  $Diam(\mathcal{W}_0)=1$. For $n>0$, we call those newly created vertices in $\mathcal{W}_n$ at $n$th iteration as active vertices. By the construction process, all active vertices are connected to those vertices existing in $\mathcal{W}_{n-1}$, so the maximum distance between any active vertex
and those vertices in $\mathcal{W}_{n-1}$ is not more than $Diam(\mathcal{W}_{n-1}) + 1$ and the maximum distance between any pair of active vertices is at most $Diam(\mathcal{W}_{n-1}) + 2$. Thus, after each iteration, the diameter of the graph increases by $2$. Hence, we have $Diam(\mathcal{W}_n) = 2n+1$ for any $n > 0$.
\end{proof}

Since for large $n$, $\ln N_n \sim n$, we have $Diam(\mathcal{W}_n) \sim \ln N_n$. Thus the diameter $Diam(\mathcal{W}_n)$ grows logarithmically with the number of vertices, indicating that the graphs $\mathcal{W}_n$ display the small-world effect.

\section{Spectra probability transition matrix and normalized Laplacian matrix}


Let $\textbf{W}_n=\textbf{W}(\mathcal{W}_n)$ denote its generalized adjacency matrix of the graphs $\mathcal{W}_n$, whose entries $W_n(i,j)$ is defined as follows: $W_n(i,j) =w_{ij}$ if vertices $i$ and $j$ are adjacent in $\mathcal{W}_n$ by an edge with weight $w_{ij}$, or $W_n(i,j) = 0$ otherwise. The diagonal strength matrix of $\mathcal{W}_n$ is $\textbf{S}_n=\textbf{S}(\mathcal{W}_n)={\rm diag} \{s_1,s_2,\ldots,s_{N_n}\}$, where the $i$th nonzero entry is the strength of vertex $i$. The transition probability matrix of $\mathcal{W}_n$, denoted by $\textbf{T}_n=\textbf{T}(\mathcal{W}_n)$, is defined by $\textbf{T}_n=\textbf{S}_n^{-1}\textbf{W}_n$, with the $(i,j)$th element $T_n(i,j)=W_n(i,j)/s_i$ accounting for the local transition probability for a walker going from vertex $i$ to $j$ in $\mathcal{W}_n$. Matrix $\textbf{T}_n$ is  an asymmetric matrix, which is similar to the normalized adjacency matrix $\textbf{P}_n=\textbf{P}(\mathcal{W}_n)$ of the graphs $\mathcal{W}_n$.
\begin{definition}\label{wdefmt1}
The normalized adjacency matrix $\textbf{P}_n$ of the graphs $\mathcal{W}_n$ is defined as
\begin{equation}\label{wmatr1}
\textbf{P}_n=\textbf{S}_n^{-\frac{1}{2}}\textbf{W}_n \textbf{S}_n^{-\frac{1}{2}}=\textbf{S}_n^{-\frac{1}{2}}\textbf{T}_n\textbf{S}_n^{\frac{1}{2}}.
\end{equation}
\end{definition}
By definition, the $(i,j)$th entry of matrix $\textbf{P}_n$ is $P_n(i,j) =\frac{W_n(i,j)}{\sqrt{s_i}\sqrt{s_j}}$. Thus, matrix $\textbf{P}_n$ is real and symmetric, and has the same set of eigenvalues as the transition probability matrix $\textbf{T}_n$. Then, in order to determine the eigenvalues of transition probability matrix $\textbf{T}_n$, we can alternatively compute those of  matrix $\textbf{P}_n$.
In addition to the transition probability matrix, we are also interested the normalized Laplacian matrix $\textbf{L}_n=\textbf{L}(\mathcal{W}_n)$ of the graphs $\mathcal{W}_n$ defined as follows.
\begin{definition}\label{wdefmt2}
The normalized Laplacian matrix of the graphs $\mathcal{W}_n$ is
\begin{equation}\label{wmatr2}
\textbf{L}_n=\textbf{I}_n-\textbf{S}_n^{-\frac{1}{2}}\textbf{W}_n\textbf{S}_n^{-\frac{1}{2}}=\textbf{I}_n-\textbf{P}_n,
\end{equation}
where $\textbf{I}_n$ is the identity matrix with the same order as $\textbf{P}_n$.
\end{definition}

For $i=1,2,\cdots,N_n$, let $\lambda_i$ and $\sigma_i$ denote the $N_n$ eigenvalues of matrices $\textbf{P}_n$ and $\textbf{L}_n$, respectively.
Since both matrices are real and symmetric, all their eigenvalues are real, which can be listed in a nondecreasing (or nonincreasing) order as: $\lambda_1 \le \lambda_2 \le \cdots \le \lambda_{N_n}$ and $\sigma_1 \ge \sigma_2 \ge \cdots \ge \sigma_{N_n}$.
From Definition~\ref{wdefmt2},  it is obvious that
\begin{equation}\label{wmatr3}
\lambda_i=1-\sigma_i
\end{equation}
holds  for all $i=1,2,\cdots,N_n$
This one-to-one correspondence implies that if one can determine the eigenvalues of one matrix, then the eigenvalues of the other matrix can be easily found.


For the convenience of the following description, we introduce a real function $f(x)$ defined to be
\begin{equation}\label{weightlm0}
f(x)=\frac{(\delta+2)x}{\delta+1}-\frac{1}{(\delta+1)(2x-1)}. \nonumber
\end{equation}

The following lemma provides the recursive relation of eigenvalues between matrices $\textbf{P}_{n-1}$ and $\textbf{P}_n$.
\begin{lemma}\label{weightlm1}\label{Recur}
If $\lambda$ is an eigenvalue of $\textbf{P}_{n}$ satisfying $\lambda \ne \pm \frac{1}{2}$, then $f(\lambda)$ is an eigenvalue of $\textbf{P}_{n-1}$, and the multiplicity of eigenvalue $f(\lambda)$ of $\textbf{P}_{n-1}$, denoted $m_{n-1}(f(\lambda))$, is the same as the multiplicity of the eigenvalue $\lambda$ of $\textbf{P}_n$, i.e. $m_{n-1}(f(\lambda))=m_n(\lambda)$.
\end{lemma}
\begin{proof}
Let $\textbf{\emph{y}}=(y_1,y_2,\cdots,y_{N_n})^{\top}$ denote the eigenvector associated with eigenvalue $\lambda$ of $\textbf{P}_n$, where the component $y_i$ corresponds to  vertex $i$ in $\mathcal{W}_n$. Then,
\begin{equation}\label{weiglm1}
\lambda \textbf{\emph{y}}=\textbf{P}_n \textbf{\emph{y}}.
\end{equation}
Let  $\mathcal{V}_n$ be the set of vertices in the graphs $\mathcal{W}_n$. Then, $\mathcal{V}_n$ can  be divided into two disjoint sets $\mathcal{V}_{n-1}$ and $\mathcal{V}_n^{'}=\mathcal{V}_n \setminus \mathcal{V}_{n-1}$, where set $\mathcal{V}_n^{'}$ contains the newly vertices created at the $n$th iteration. For all vertices in $\mathcal{V}_n$, we label those in  $\mathcal{V}_{n-1}$ from $1$ to $N_{n-1}$, while label the vertices in $\mathcal{V}_n^{'}$ from $N_{n-1}+1$ to $N_{n}$.

For an old vertex $o \in \mathcal{V}_{n-1}$ that was generated before iteration $n$, the row in Eq.~\eqref{weiglm1} corresponding to component $y_o$ can be written as
\begin{equation}\label{weiglm2}
\lambda y_o = \sum_{i=1}^{N_n}P_n(o,i) y_i.
\end{equation}
By construction, for each newly created vertex $k$ linked to $o$, there is only one vertex $h \in \mathcal{V}_n^{'}$ that is simultaneously adjacent to both $k$ and $o$. According to Eq.~\eqref{weiglm1}, the characteristic equations associated vertices $y_k$ and $y_h$ can be expressed as
\begin{equation}\label{weiglm3}
\lambda y_k = \sum_{i=1}^{N_n}P_n(k,i) y_i=P_n(k,o)y_o+P_n(k,h)y_{h}
\end{equation}
and
\begin{equation}\label{weiglm4}
\lambda y_h = \sum_{i=1}^{N_n}P_n(h,i) y_i=P_n(h,o)y_o+P_n(h,k)y_{k},
\end{equation}
respectively.
By definition of matrix $\textbf{P}_n$, we have $P_n(h,o)=P_n(k,o)$ and
$P_n(h,k)=P_n(k,h)=\frac{1}{2}$.
Thus, Eqs.~\eqref{weiglm3} and~\eqref{weiglm4} can be rewritten as
\begin{equation}\label{weiglm301}
\lambda y_k = P_n(k,o)y_o+\frac{1}{2}y_{h}
\end{equation}
and
\begin{equation}\label{weiglm401}
\lambda y_h = P_n(k,o)y_o+\frac{1}{2}y_{k},
\end{equation}
respectively. From Eq.~\eqref{weiglm401}, we obtain
\begin{equation}\label{weiglm402}
y_h=\frac{2P_n(k,o)y_o+y_{k}}{2\lambda}.
\end{equation}
Plugging Eq.~\eqref{weiglm402} into Eq.~\eqref{weiglm301} yields
\begin{equation}\label{weiglm7}
(2\lambda+1)\bigg(2P_n(k,o)y_o-(2\lambda-1)y_k \bigg)=0, \nonumber
\end{equation}
which shows that
\begin{equation}\label{weiglm8}
y_k=\frac{P_n(k,o)}{\lambda-\frac{1}{2}}y_o
\end{equation}
holds for $\lambda \ne \pm\frac{1}{2}$.
By Definition~\ref{wdefmt1}, $\textbf{P}_n$ is real and symmetric. Thus, for any vertex $k$ linked to $o$, we have
\begin{equation}\label{weiglm801}
P_n(o,k)=P_n(k,o).
\end{equation}
Substituting Eq.~\eqref{weiglm8} into Eq.~\eqref{weiglm2} and considering the Eq.~\eqref{weiglm801}  yields
\begin{equation}\label{weiglm9}
\sum_{i=1}^{N_{n-1}}P_n(o,i)y_i = \bigg( \lambda-\sum_{j=N_{n-1}+1}^{N_n}\frac{P_n^2(o,j)}{\lambda-\frac{1}{2}} \bigg) y_o.
\end{equation}
By construction of the graph  and combining Eqs.~\eqref{weightstrdis3} and~\eqref{wweigdis1}, we have
\begin{equation}\label{weiglm10}
P_n(o,i)=\frac{W_n(o,i)}{\sqrt{s_o(n)}\sqrt{s_i(n)}}=\frac{(\delta+1)W_{n-1}(o,i)}{\sqrt{(\delta+2)s_o(n-1)}\sqrt{(\delta+2)s_i(n-1)}}=\frac{\delta+1}{\delta+2}P_{n-1}(o,i)
\end{equation}
for $i=1,2,\cdots,N_{n-1}$,
and
\begin{equation}\label{weiglm11}
\sum_{j=N_{n-1}+1}^{N_n}P_n^2(o,j) = \frac{1}{\delta+2}.
\end{equation}
Substituting Eqs.~\eqref{weiglm10} and~\eqref{weiglm11} into Eq.~\eqref{weiglm9} gives
\begin{equation}\label{weiglm12}
\sum_{j=1}^{N_{n-1}}P_{n-1}(o,j)y_j = \left(\frac{(\delta+2)\lambda}{\delta+1}-\frac{1}{(\delta+1)(2\lambda-1)}\right)y_o=f(\lambda)y_o,
\end{equation}
which indicates that if $\lambda$ is an eigenvalue of matrix $\textbf{P}_n$, then $f(\lambda)$ is an eigenvalue of $\textbf{P}_{n-1}$, whose corresponding  eigenvector is $(y_1,y_2,\cdots,y_{N_{n-1}})^{\top}$.

Let $\widetilde{\lambda}=f(\lambda)$ be an eigenvalue of matrix $\textbf{P}_{n-1}$.
Because $\textbf{P}_{n-1}$ is a real and symmetrical matrix, each eigenvalue $\widetilde{\lambda}$ has $m_{n-1}(\widetilde{\lambda})$ linearly independent eigenvectors.
Suppose $\widetilde{\textbf{\emph{y}}}^{\top}=(y_1,y_2,\cdots,y_{N_{n-1}})^{\top}$ is an arbitrary eigenvector corresponding to $\widetilde{\lambda}$,
i.e. $\textbf{P}_{n-1}\widetilde{\textbf{\emph{y}}}=\widetilde{\lambda} \widetilde{\textbf{\emph{y}}}$,
then vector $\textbf{\emph{y}}^{\top}=(y_1,y_2,\cdots,y_{N_{n-1}},y_{N_{n-1}+1},\cdots,y_{N_{n}})^{\top}$ is an eigenvector corresponding to eigenvalue $\lambda$ of matrix $\textbf{P}_{n}$ if and only if its component $y_i$, $i=N_{n-1}+1, N_{n-1}+2, \ldots, N_{n}$,  can be expressed by Eq.~\eqref{weiglm8}. Thus, the number of linearly independent eigenvectors of $\lambda$ is the same as that of $\widetilde{\lambda}$, which means that $m_{n-1}(f(\lambda))=m_n(\lambda)$. This completes the proof.

\end{proof}

Lemma~\ref{Recur} indicates that except $\lambda \ne \pm\frac{1}{2}$, all eigenvalues $\lambda$ of matrix $\textbf{P}_n$ can be derived from those of $\textbf{P}_{n-1}$. However, it is easy to check that both $\frac{1}{2}$ and $-\frac{1}{2}$ are eigenvalues of $\textbf{P}_n$, and their multiplicity can be determined by the following lemma.
\begin{lemma}\label{weightlm2}
The multiplicity of eigenvalue $\frac{1}{2}$ and $-\frac{1}{2}$ of $\textbf{P}_n$ is
\begin{equation}\label{weightspa1}
m_n \left(\frac{1}{2} \right)=\frac{3(\delta+1)((\delta+4)^{n-1}-1)}{\delta+3}
\end{equation}
and
\begin{equation}\label{weightspa2}
m_n \left(-\frac{1}{2} \right)=\frac{3(\delta+4)^{n+1}+2\delta+3}{\delta+3},
\end{equation}
respectively.
\end{lemma}
\begin{proof}
Let $r(\textbf{M})$ denote the rank of matrix $\textbf{M}$, then the multiplicity of eigenvalue $\frac{1}{2}$ can be evaluated by
\begin{equation}\label{weiglm15}
m_n\left(\frac{1}{2}\right)=N_n-r\left(\textbf{P}_n-\frac{1}{2}\textbf{I}_n\right).
\end{equation}
where $\textbf{I}_n$ is the identity matrix with the same order as $\textbf{P}_n$.
For simplicity, let $\alpha=\mathcal{V}_{n-1}$ and $\beta=\mathcal{V}_n^{'}$. Then, $\mathcal{V}_{n}=\alpha \cup \beta$ 
and matrix $\textbf{P}_n-\frac{1}{2}\textbf{I}_n$ can be expressed in the following block form.

\begin{equation}\label{weiglm16}
\textbf{P}_{n}-\frac{1}{2}\textbf{I}_n=\left[\begin{array}{cccc}
\widetilde{\textbf{P}}_{\alpha,\alpha} & \widetilde{\textbf{P}}_{\alpha, \beta} \\
\widetilde{\textbf{P}}_{\beta,\alpha} & \widetilde{\textbf{P}}_{\beta, \beta}
\end{array}\nonumber
\right]
\end{equation}
where $\widetilde{\textbf{P}}_{\alpha,\alpha}$ is an $N_{n-1} \times N_{n-1}$ matrix and $\widetilde{\textbf{P}}_{\alpha, \beta}$ is an $N_{n-1} \times (N_n-N_{n-1})$ matrix. The block matrix $\widetilde{\textbf{P}}_{\alpha, \beta}=\widetilde{\textbf{P}}_{\beta,\alpha}^{\top}$ takes the form
\begin{equation}\label{weiglm17}
\setlength{\arraycolsep}{1.5pt}
\widetilde{\textbf{P}}_{\alpha,\beta}=\left[\begin{array}{ccccccccccc}
  p_1&\cdots&p_1 & & & & & & &\\
  & & & p_2& \cdots&p_2 & & & & \\
  & & & & & &\ddots & & &  \\
   & & & & & & & p_{_{N_{n-1}}}&\cdots&p_{_{N_{n-1}}}
\end{array}
\right],\nonumber
\end{equation}
where all unmarked entries are zeros.
For arbitrary $1 \le i \le N_{n-1}$, $p_i \ne 0$ and the repeating times of each $p_i$ are even, equaling the number of the new
neighbors for vertex $i$. $\widetilde{\textbf{P}}_{\beta, \beta}$ is a $(N_n-N_{n-1}) \times (N_n-N_{n-1})$ matrix and has the form
$\widetilde{\textbf{P}}_{\beta,\beta}={\rm diag}(\textbf{A},\textbf{A},\cdots,\textbf{A})$
with
$\textbf{A}=\left[\begin{array}{cccc}
-\frac{1}{2} & \frac{1}{2} \\
\frac{1}{2} & -\frac{1}{2}
\end{array}
\right]$.

Performing some identical elementary column operations on $\widetilde{\textbf{P}}_{\alpha,\beta}$ and
$\widetilde{\textbf{P}}_{\beta,\beta}$, we obtain
\begin{equation}\label{weiglm20}
\setlength{\arraycolsep}{1.5pt}
r(\widetilde{\textbf{P}}_{\alpha,\beta})=r\left(\left[\begin{array}{cccccccccccc}
  p_1&0 &\cdots&0 & & & & & & &\\
  & & & & p_2& 0& \cdots& 0 & & & & \\
  & & & & & & &\ddots & & &  \\
   & & & & & & & & p_{_{N_{n-1}}}& 0 & \cdots&0
\end{array}
\right]\right)\nonumber
\end{equation}
and
\begin{equation}\label{weiglm21}
r(\widetilde{\textbf{P}}_{\beta,\beta})=r\left({\rm diag} (\textbf{B},\textbf{B},\cdots,\textbf{B}) \right) \nonumber
\end{equation}
where $\textbf{B}=\left[\begin{array}{cccc}
0 & \frac{1}{2} \\
0 & -\frac{1}{2}
\end{array}
\right].$ 
Thus, we have
\begin{equation}\label{weiglm23}
r\bigg(\left[\begin{array}{c}
\widetilde{\textbf{P}}_{\alpha,\beta}\\
\widetilde{\textbf{P}}_{\beta,\beta}
\end{array} \right]\bigg)=\frac{N_{n}-N_{n-1}}{2}+N_{n-1}=\frac{N_{n}+N_{n-1}}{2}.
\end{equation}
By performing similar elementary column operations on matrix $\textbf{P}_n-\frac{1}{2}\textbf{I}_n$, we obtain
\begin{equation}\label{weiglm24}
\setlength{\arraycolsep}{1.0pt}
r\left(\textbf{P}_n-\frac{1}{2}\textbf{I}_n\right)=r\left(\left[\begin{array}{cccc|cccccccccc}
& & & &  p_1&0&\cdots & 0& & & & & &  \\
& &\widetilde{\textbf{P}}_{\alpha,\alpha} & & & & & & \ddots& & & & &  \\
& & & & & & & & & \cdots & & & &  \\
& & & & & & & &  & &p_{_{N_{n-1}}} &0 & \cdots & 0\\
\hline
p_1 & & &  &0&\frac{1}{2} & & & &  & & & & \\
\vdots & & & &0&-\frac{1}{2}& & & & & & & & \\
\vdots & & & & & & \ddots& & & & & & & \\
p_1 & & & & & & & \ddots & &  & & & &  \\
&\ddots & & & & & & &0 & \frac{1}{2}& & & &  \\
& &\vdots & & & & & &0 & -\frac{1}{2} & & & &  \\
& &\vdots & & & & & & & & \ddots& & &  \\
& & &p_{_{N_{n-1}}}& & & & & & & &  \ddots& & \\
& & & \vdots& & & & & & &  & &0&\frac{1}{2} \\
& & &p_{_{N_{n-1}}}  & & & & & & & & &0& -\frac{1}{2}
\end{array}
\right]\right).
\end{equation}
Combining Eqs.~\eqref{weiglm23} and~\eqref{weiglm24}, we have
\begin{equation}\label{weiglm25}
r\left(\textbf{P}_n-\frac{1}{2}\textbf{I}_n\right)=r\bigg(\left[\begin{array}{c}
\widetilde{\textbf{P}}_{\alpha,\beta}\\
\widetilde{\textbf{P}}_{\beta,\beta}
\end{array} \right]\bigg)+N_{n-1}=\frac{N_{n}+3N_{n-1}}{2}.
\end{equation}
Substituting Eq.~\eqref{weiglm25} into Eq.~\eqref{weiglm15} yields
\begin{equation}\label{weiglm26}
m_n\bigg(\frac{1}{2}\bigg)=\frac{N_n-3N_{n-1}}{2}=\frac{3(\delta+1)((\delta+4)^{n-1}-1)}{\delta+3}.\nonumber
\end{equation}

We proceed to compute the multiplicity $m_n\left(-\frac{1}{2}\right)$ of eigenvalue $-\frac{1}{2}$ of matrix $\textbf{P}_n$, which satisfies
\begin{equation}\label{weiglm27}
m_n\left(\frac{1}{2}\right)+m_n\left(-\frac{1}{2}\right)+m_n\left(\lambda \ne \pm\frac{1}{2}\right)=N_n,
\end{equation}
where $m_n\left(\lambda \ne \pm\frac{1}{2}\right)$ represents the sum of multiplicity of all eigenvalues of $\textbf{P}_n$ excluding $\frac{1}{2}$ and $-\frac{1}{2}$. Since $f\left(-\frac{1}{2}\right)=-\frac{1}{2}$, we have the following relation between $m_n\left(\lambda \ne \pm\frac{1}{2}\right)$ and $m_{n-1}\left(-\frac{1}{2}\right)$:
\begin{equation}\label{weiglm28}
m_{n-1}\left(-\frac{1}{2}\right)+m_n\left(\lambda \ne \pm\frac{1}{2}\right)=2N_{n-1}.
\end{equation}
Combining Eqs.~\eqref{weiglm27} and~\eqref{weiglm28} and using $m_0(-\frac{1}{2})=2$, we get
\begin{equation}\label{weiglm29}
m_n\left(-\frac{1}{2}\right)=\frac{3(\delta+4)^{n+1}+2\delta+3}{\delta+3}.\nonumber
\end{equation}
The proof is completed.
\end{proof}

Before giving our main result of this section, we introduce two more functions $f_1(x)$ and $f_2(x)$ defined as
\begin{equation}\label{weightspf1}
f_1(x)=\frac{2+2x+\delta+2x\delta-\sqrt{(2+\delta+2x(1+\delta))^2-8(\delta+2)(x\delta+x-1)}}{4(\delta+2)}\nonumber
\end{equation}
and
\begin{equation}\label{weightspf2}
f_2(x)=\frac{2+2x+\delta+2x\delta+\sqrt{(2+\delta+2x(1+\delta))^2-8(\delta+2)(x\delta+x-1)}}{4(\delta+2)},\nonumber
\end{equation}
respectively. Lemma~\ref{weightlm1} indicates that if $\tilde{\lambda}$ is an eigenvalue of matrix $\textbf{P}_{n-1}$, then both $f_1(\tilde{\lambda})$ and $f_2(\tilde{\lambda})$ are eigenvalues of matrix $\textbf{P}_{n}$.


\begin{theorem}\label{weightth1}
Let $\Delta_{n-1}=\{\underbrace{\lambda_1,\cdots,\lambda_1}_\text{$n_1$},\underbrace{\lambda_2,\cdots,\lambda_2}_\text{$n_2$},\cdots,\underbrace{\lambda_k,\cdots,\lambda_k}_\text{$n_k$}\}$ be the set of all eigenvalues of matrix $\textbf{P}_{n-1}$, where $\lambda_k=-\frac{1}{2}$ and $\lambda_i \neq \lambda_j$ for all $i \neq j$. Then the set of all eigenvalues eigenvalues of matrix $\textbf{P}_{n}$ is
\begin{equation}\nonumber
\begin{split}
\Delta_n&= \bigg \{
\underbrace{f_1(\lambda_1),\ldots,f_1(\lambda_1)}_\text{$n_1$},
\underbrace{f_2(\lambda_1),\ldots,f_2(\lambda_1)}_\text{$n_1$}, \ldots,
\underbrace{f_1(\lambda_{k-1}),\ldots,f_1(\lambda_{k-1})}_\text{$n_{k-1}$},\\
&\quad \underbrace{f_2(\lambda_{k-1}),\ldots,f_2(\lambda_{k-1})}_\text{$n_{k-1}$},
\underbrace{f_1(\lambda_k),\ldots,f_1(\lambda_k)}_\text{$n_k$},\underbrace{-\frac{1}{2},\ldots,-\frac{1}{2}}_\text{$m_n\left(-\frac{1}{2}\right)$}, \underbrace{\frac{1}{2},\ldots,\frac{1}{2}}_\text{$m_n\left(\frac{1}{2}\right)$}\bigg \},
\end{split}
\end{equation}
where $m_n\left(-\frac{1}{2}\right)=\frac{3(\delta+1)((\delta+4)^{n-1}-1)}{\delta+3}$ and $m_n\left(\frac{1}{2}\right)=\frac{3(\delta+1)((\delta+4)^{n-1}+1)}{\delta+3}$.
\end{theorem}
\begin{proof}
This Theorem is a direct consequence of Lemma~\ref{weightlm1} and~\ref{weightlm2}.
\end{proof}

Note that the eigenvalue set of graph $\mathcal{W}_0$ is $\Delta_0=\left\{-\frac{1}{2} , -\frac{1}{2} , 1 \right\}$. By recursively applying the result of above theorem, we can obtain the eigenvalues of the transition matrix $\textbf{P}_n$ for the graphs $\mathcal{W}_n$ for any $n$.

On the other hand, combining Eq.~\eqref{wmatr3} and Theorem~\ref{weightth1}, we can also obtain the eigenvalues of the normalized Laplacian matrix $\textbf{L}_n$ for  $\mathcal{W}_n$. For this purpose, we define two functions $g_1(x)$ and $g_2(x)$:
\begin{equation}\label{weightspg1}
g_1(x)=\frac{\delta +2 x+4+2 \delta  x-\sqrt{(2 \delta  x+\delta +2 x+4)^2-4 (2 \delta +4)
   (\delta  x+x)}}{2 (2 \delta +4)},\nonumber
\end{equation}
\begin{equation}\label{weightspg2}
g_2(x)=\frac{\delta +2 x+4+2 \delta  x+\sqrt{(2 \delta  x+\delta +2 x+4)^2-4 (2 \delta +4)
   (\delta  x+x)}}{2 (2 \delta +4)}.\nonumber
\end{equation}
The  following theorem gives the recursive relation of the eigenvalues between $\textbf{L}_{n-1}$ and $\textbf{L}_n$.
\begin{theorem}\label{coro1}
Let $\Omega_{n-1}=\{\underbrace{\sigma_1,\cdots,\sigma_1}_\text{$n_1$},\underbrace{\sigma_2,\cdots,\sigma_2}_\text{$n_2$},\cdots,\underbrace{\sigma_k,\cdots,\sigma_k}_\text{$n_k$}\}$, be the set of eigenvalues of $\textbf{L}_{n-1}$, where $\sigma_k=\frac{3}{2}$ and $\sigma_i \neq \sigma_j$ for all $i \neq j$. Then the set of  eigenvalues of $\textbf{L}_{n}$ is
\begin{equation}\nonumber
\begin{split}
\Omega_{n}&=\bigg \{
\underbrace{g_1(\sigma_1),\ldots,g_1(\sigma_1)}_\text{$n_1$},
\underbrace{g_2(\sigma_1),\ldots,g_2(\sigma_1)}_\text{$n_1$}, \ldots,
\underbrace{g_1(\sigma_{k-1}),\ldots,g_1(\sigma_{k-1})}_\text{$n_{k-1}$},\\
&\quad \underbrace{g_2(\sigma_{k-1}),\cdots,g_2(\sigma_{k-1})}_\text{$n_{k-1}$},
\underbrace{g_1(\sigma_k),\ldots,g_1(\sigma_k)}_\text{$n_k$},\underbrace{\frac{3}{2},\ldots,\frac{3}{2}}_\text{$m_n\left(\frac{3}{2}\right)$}, \underbrace{\frac{1}{2},\ldots,\frac{1}{2}}_\text{$m_n\left(\frac{1}{2}\right)$} \bigg \} .
\end{split}
\end{equation}
where $m_n\left(\frac{3}{2}\right)=\frac{3(\delta+1)((\delta+4)^{n-1}-1)}{\delta+3}$ and $m_n\left(\frac{1}{2}\right)=\frac{3(\delta+1)((\delta+4)^{n-1}+1)}{\delta+3}$.
\end{theorem}

\section{Applications of eigenvalues}

In this section, we show how to apply the above-obtained eigenvalues and their properties to evaluate some related quantities for the weighted graphs $\mathcal{W}_n$, including mean hitting time and weighted counting of spanning trees.

\subsection{Mean hitting time}

The probability transition matrix $\textbf{T}(\mathcal{G})$ of a weighted graph $\mathcal{G}$ depicts the process of a biased random walk running on the graph. During the process of the random walk, at each time step, the walker moving from vertex $i$ to one of its neighboring vertices $j$ with the probability $\frac{w_{ij}}{s_i}$, which constitutes the $ij$th entry of probability transition matrix $\textbf{T}(\mathcal{G})$. A lot of interesting quantities related to this random walk are encoded in probability transition matrix. Here we are concerned with the mean hitting time of random walks, which reflects the structural and weighted properties of the whole graph.

Let $H_{ij}$ denote the hitting time (also called first-passage time~\cite{Re01,NoRi04,CoBeTeVoKl07}) from vertex $i$ to vertex $j$ in graph $\mathcal{G}$, which is the expected time taken by a random walker to first reach vertex $j$ starting from vertex $i$. Let $\pi = (\pi_1,\pi_2, \cdots,\pi_{N})$ denote the stationary distribution for the random walk on $\mathcal{G}$~\cite{LaLoPa96,AlFi02}, where $\pi_i = \frac{s_i}{\sum_{i=1}^{N}s_i}$, satisfying the relations $\pi^{\top}\textbf{T} = \pi^{\top}$ and  $\sum_{i=1}^{N}\pi_i = 1$. Then, the mean hitting time $H$ is defined as the expected time for a random walker going from a node $i$ to another node $j$, selected randomly from all nodes in $\mathcal{G}$ according to the stationary distribution~\cite{LaLoPa96,AlFi02}, that is
\begin{equation}\label{rtat1}
H =\sum_{j=1}^{N}\pi_jH_{ij}.\nonumber
\end{equation}
The quantity $H$ does not depend on the starting node $i$, which can be expressed in terms of the nonzero eigenvalues of normalized Laplacian matrix $\textbf{L}$, given by~\cite{AlFi02,LeLo02}
\begin{equation}\label{rtat3}
H=\sum_{i=1}^{N-1}\frac{1}{\sigma_{i}}.\nonumber
\end{equation}

Mean hitting time  has found many applications in different areas~\cite{Hu14}. For example, it can be applied  to measure the efficiency of user navigation through the World Wide Web~\cite{LeLo02}, as well as the efficiency of robotic surveillance in network environments~\cite{PaAgBu15}.


\begin{theorem}\label{theohit}
Let $H_n$ be the mean hitting time for random walk in the weighted graphs $\mathcal{W}_n$. Then,
\begin{equation}\label{weigapp0}
H_n=\frac{4\delta(\delta+1)^{n+1}+24(\delta+4)^n(\delta+1)^{n+1}-12(\delta+2)(\delta+4)^n}{3\delta(\delta+4)(\delta+1)^n}.
\end{equation}
\end{theorem}
\begin{proof}
Based on the previously obtained result~\cite{LaLoPa96,AlFi02}, the mean hitting time $H_n$ for graph $\mathcal{W}_n$ can be expressed in terms of nonzero eigenvalues of matrix $\textbf{L}_n$ as
\begin{equation}\label{weigapp1}
H_n=\sum_{\sigma \in \Omega_n\setminus \{0\}}\frac{1}{\sigma}.
\end{equation}
In order to determine $H_n$, we divide $\Omega_n$ into two subsets $\Omega_n^{1}$ and $\Omega_n^{2}$ satisfying $\Omega_n=\Omega_n^{1}\cup \Omega_n^{2}$, where $\Omega_n^{1}$ contains all the nonzero eigenvalues that are generated from $\Omega_{n-1}$ by functions $g_1(x)$ and $g_2(x)$, while $\Omega_n^{2}$ contains all the other eigenvalues in  $\Omega_n$. It is obvious that $\Omega_n^{2}$ contains all the eigenvalues $\frac{1}{2}$ and a part of eigenvalues $\frac{3}{2}$ in $\Omega_n$. Then, Eq.~\eqref{weigapp1} can be rewritten as
\begin{equation}\label{weigapp2}
H_n=\sum_{\sigma \in \Omega_n^1\setminus \{0\}}\frac{1}{\sigma}+\sum_{\sigma \in \Omega_n^2}\frac{1}{\sigma}.
\end{equation}
We denote the two sums on the rhs of Eq.~\eqref{weigapp2} by $H_n^{(1)}$ and $H_n^{(2)}$, respectively. The first sum can be evaluated as
\begin{equation}\label{weigapp3}
H_n^{(1)}=\sum_{\sigma \in \Omega_n^1\setminus \{0\}}\frac{1}{\sigma}=\sum_{\sigma \in \Omega_{n-1}\setminus \{0\}}\left(\frac{1}{g_1(\sigma)}+\frac{1}{g_2(\sigma)}\right)+\frac{1}{g_2(0)}.
\end{equation}
According to Vieta's formulas, we have
\begin{equation}\label{weigapp4}
g_1(\sigma)+g_2(\sigma)=\frac{(2\delta+2)\sigma+\delta+4}{2(\delta+2)},
\end{equation}
\begin{equation}\label{weigapp5}
g_1(\sigma)\cdot g_2(\sigma)=\frac{(\delta+1)\sigma}{2(\delta+2)},
\end{equation}
and
\begin{equation}
g_2(0)=\frac{\delta+4}{2\delta+4}.
\end{equation}
Then, Eq.~\eqref{weigapp3} can be rephrased as
\begin{equation}\label{weigapp6}
H_n^{(1)}=\sum_{\sigma \in \Omega_{n-1}\setminus \{0\}} \left(\frac{\delta+4}{(\delta+1)\sigma}+2\right)+\frac{2\delta+4}{\delta+4}.
\end{equation}
The second sum in Eq.~\eqref{weigapp2} can be determined as
\begin{equation}\label{weigapp7}
H_n^{(2)}=\sum_{\sigma \in \Omega_n^2}\frac{1}{\sigma}=2m_n\left(\frac{1}{2}\right)+\frac{2}{3}\left(m_n\left(-\frac{1}{2}\right)-m_{n-1}\left(-\frac{1}{2}\right)\right).
\end{equation}
Combining Eqs.~\eqref{weigapp6} and~\eqref{weigapp7}, we obtain the recursive relation for $H_n$:
\begin{equation}\label{weigapp8}
H_n=\frac{\delta+4}{\delta+1}H_{n-1}+8(\delta+4)^{n-1}-\frac{4}{\delta+4}.
\end{equation}
Using the initial condition $H_0 = \frac{4}{3}$, Eq.~\eqref{weigapp8} is solved to obtain
\begin{equation}\label{weigapp9}
H_n=\frac{4\delta(\delta+1)^{n+1}+24(\delta+4)^n(\delta+1)^{n+1}-12(\delta+2)(\delta+4)^n}{3\delta(\delta+4)(\delta+1)^n}.
\end{equation}
This completes the proof.
\end{proof}

From Theorem \ref{theohit}, we can  see that for large $n$ (i.e. $n \to \infty $), the dependence of $H_n$ on the order $N_n$ of the graphs $\mathcal{W}_n$ is $H_n \sim N_n$, which indicates that the mean hitting time $H_n$ grows linearly with the number of vertices.

\subsection{Weighted counting of spanning trees}

For a weighted graph $\mathcal{G}$, let $\Upsilon (\mathcal{G})$ represent the set of its spanning trees. For  a tree $\zeta$ in $\Upsilon (\mathcal{G})$, its weight $w(\zeta)$ is defined to be the product of weights of all edges in $\zeta$: $w(\zeta)=\prod_{e \in \zeta}w_e$, where $w_e$ is the weight of edge $e$ in  $\zeta$.  Let $\tau(\mathcal{G})$ be the weighted counting of spanning trees of $\mathcal{G}$, which is defined  by $\tau(\mathcal{G})=\sum_{\zeta \in \Upsilon (G)}w(\zeta)$. It has been shown~\cite{Ch11,ChXuYa14} that $\tau(\mathcal{G})$ can be expressed in terms of the non-zero eigenvalues of normalized Laplacian matrix of $\mathcal{G}$ and the strength of all vertices as $\tau(\mathcal{G})=\frac{\prod_{i=1}^{N}s_i \prod_{i=2}^{N}\sigma_i}{\sum_{i=1}^{N}s_i}$.


The weighted counting of spanning trees is an important graph invariant, which is useful in identifying important vertices in weighted networks~\cite{QiFuLuZh15}.  In the sequel,
we will apply the obtained eigenvalues to determine this invariant in the graphs $\mathcal{W}_n$.



\begin{theorem}\label{theowwcst}
Let $\tau(\mathcal{W}_n)$ be the weighted counting of spanning trees in the graphs $\mathcal{W}_n$. Then, for all $n \geq 0$,
\begin{equation}\label{wwscpt1}
\tau(\mathcal{W}_n)=3^{\frac{3 (\delta +4)^n+\delta }{\delta +3}} (\delta +1)^{\frac{2 \left(3 \left((\delta +4)^n-1\right)+\delta  (\delta +3) n\right)}{(\delta +3)^2}}.
\end{equation}
\end{theorem}

\begin{proof}
According to previous result~\cite{Ch11,ChXuYa14}, $\tau(\mathcal{W}_n)$ can be expressed as
\begin{equation}\label{wwscpt2}
\tau(\mathcal{W}_n)=\frac{\displaystyle{\prod_{i=1}^{N_n}s_i(n)} \prod_{\sigma \in \Omega_{n}\setminus \{0\}}\sigma}{\displaystyle{\sum_{i=1}^{N_n}s_i(n)}},
\end{equation}
where $s_i(n)$ denotes the strength of vertex $i$ in  $\mathcal{W}_n$. The three terms (one sum term in the denominator and two product terms in the numerator) on the rhs of Eq.~\eqref{wwscpt2} can be evaluated as follows.

For the sum term in the denominator, we have
\begin{equation}\label{wwscpt3}
\sum_{i=1}^{N_n}s_i(n)=2W_n=6(\delta+4)^n.
\end{equation}
Let $S_n$ denote the product term $\prod_{i=1}^{N_n}s_i(n)$. By construction, $S_n$ obeys the following recursive relation:
\begin{equation}\label{wwscpt4}
S_n=(\delta+2)^{N_{n-1}}S_{n-1}\cdot 2^{n_v(n)}.
\end{equation}
With the initial condition $S_0=8$, Eq.~\eqref{wwscpt4} is solved to yield
\begin{equation}\label{wwscpt5}
S_n=8^{\frac{2 (\delta +4)^n+\delta +1}{\delta +3}} (\delta +2)^{\frac{3 \left(2 \left((\delta +4)^n-1\right)+(\delta +1) (\delta +3)n\right)}{(\delta+3)^2}}.
\end{equation}
Let $M_n$ denote  the product term $\prod_{\sigma \in \Omega_{n}\setminus \{0\}}\sigma$. Then,
\begin{equation}\label{wwscpt6}
M_n=\prod_{\sigma \in \Omega_{n}\setminus \{0\}}\sigma=\prod_{\sigma \in \Omega_{n}^{1}\setminus \{0\}}\sigma \cdot \prod_{\sigma \in \Omega_{n}^{2}}\sigma
\end{equation}
Let $M_n^{(1)}$ and $M_n^{(2)}$ represent, respectively, the two products on the rhs  of Eq.~\eqref{wwscpt6}. $M_n^{(1)}$ can be computed by
\begin{equation}\label{wwscpt7}
M_n^{(1)}=\prod_{\sigma \in \Omega_{n}^{1}\setminus \{0\}}\sigma=g_2(0)\cdot \prod_{\sigma \in \Omega_{n-1}\setminus \{0\}}(g_1(\sigma)\cdot g_2(\sigma)).
\end{equation}
Inserting Eq.~\eqref{weigapp5} into Eq.~\eqref{wwscpt7} results in the following recursive relation:
\begin{equation}\label{wwscpt8}
M_n^{(1)}=\frac{\delta+4}{2\delta+4}\cdot \bigg(\frac{\delta+1}{2\delta+4}\bigg)^{N_{n-1}-1} \cdot M_{n-1}.
\end{equation}
In addition, $M_n^{(2)}$ can be expressed as
\begin{equation}\label{wwscpt9}
M_n^{(2)}=\left(\frac{1}{2}\right)^{m_n\left(\frac{1}{2}\right)} \cdot \left(\frac{3}{2}\right)^{m_n\left(-\frac{1}{2}\right)-m_{n-1}\left(-\frac{1}{2}\right)}.
\end{equation}
Using $M_n=M_n^{(1)}\cdot M_n^{(2)}$ and $M_0=\frac{9}{4}$, we  obtain
\begin{equation}\label{wwscpt10}
M_n=2^{-\frac{2(\delta+3(\delta+4)^n)}{3+\delta}}\cdot 3^{\frac{3+2\delta+3(\delta+4)^n}{3+\delta}}\cdot \bigg(\frac{\delta+4}{\delta+2}\bigg)^n \cdot \bigg(\frac{\delta+1}{\delta+2}\bigg)^{\frac{2(3((\delta+4)^n-1)+\delta(\delta+3)n)}{(\delta+3)^2}}.
\end{equation}
Combining the above obtained results, we get
\begin{equation}\label{wwscpt11}
\tau(\mathcal{W}_n)=3^{\frac{3 (\delta +4)^n+\delta }{\delta +3}} (\delta +1)^{\frac{2 \left(3 \left((\delta +4)^n-1\right)+\delta  (\delta +3) n\right)}{(\delta +3)^2}}.
\end{equation}
Hence the proof.
\end{proof}

In fact, $\tau(\mathcal{W}_n)$ can also be evaluated by direct enumeration. It is easy to observe that all the spanning trees in $\Upsilon(\mathcal{W}_n)$ have the same weight. Note that there are $L_{\triangle}(n)=\frac{3(\delta+4)^n+\delta}{\delta+3}$ triangles in  $\mathcal{W}_n$, moreover, the three edges of each triangle have identical weight. By construction, we can obtain the weight distribution of edges of the $L_{\triangle}(n)$ triangles in $\mathcal{W}_n$: the number of triangles with edge weight (for each edge) $1$, $(1+\delta)^2$, $\cdots$, $(1+\delta)^{n-2}$, $(1+\delta)^{n-1}$, $(1+\delta)^n$, equals, respectively, $3(4+\delta)^{n-1}$, $3(4+\delta)^{n-2}$, $\cdots$, $3(4+\delta)$, $3$, and $1$.  Then, $\tau(\mathcal{W}_n)$ can be alternatively expressed as
\begin{align}\label{wwscpt12}
\tau(\mathcal{W}_n) & = 3^{L_{\Delta}(n)} \cdot (1+\delta)^{2n} \cdot \prod_{i=1}^{n} (1+\delta)^{2(n-i) \cdot 3(4+\delta)^{i-1}} \nonumber\\
& =  3^{\frac{3 (\delta +4)^n+\delta }{\delta +3}} (\delta +1)^{\frac{2 \left(3 \left((\delta +4)^n-1\right)+\delta  (\delta +3) n\right)}{(\delta +3)^2}},\nonumber
\end{align}
which is in full agreement with Eq.~\eqref{wwscpt1},  indicating the validity of our computation on the eigenvalues and their multiplicities for related matrix of  $\mathcal{W}_n$.

It deserves to mention that since the weights of all edges in $\mathcal{W}_n$ are integer, each edge $e$ with weight $w_e$ can be looked upon as $w_e$ parallel edges~\cite{QiFuLuZh15}, each having unit weigh and being linked to the two endvertices of edge $e$. Then,  every spanning tree $\zeta$ with weight $w(\zeta)$ in $\Upsilon (\mathcal{W}_n)$ can be considered as  $w(\zeta)$ trees with unit weight and identical topological structure, and  the weighted counting of spanning trees $\tau(\mathcal{W}_n)$ can be regarded as $\tau(\mathcal{W}_n)$ spanning trees in $\mathcal{W}_n$, each having  unit weight.

\section{Conclusion}

A strong advantage of modeling real networks using graph products is that one can theoretically analyze the structural and spectral characteristics of the resulting graphs. In this paper, we have extended the corona product of binary graphs to weighted cases. Based on the extended corona product and the weight reinforcement mechanism in real systems, we have proposed a model for heterogeneous weighted networks. We have presented a detailed analysis for relevant properties of the model. The obtained analytical expressions indicate that the resulting weighted networks exhibit power-law distribution of node strength, node degree, and edge weight; moreover, the networks have small diameter and high clustering coefficient. Thus, the model can well mimic the properties of real weighted networks.

Moreover, we have found all the eigenvalues as well as their multiplicities of the transition probability matrix for random walks on the proposed weighed networks. Based on these eigenvalues, we have further evaluated the mean hitting time for random walks on the weighed networks, which grows linearly with the number of nodes. We have also derived the weighted counting of spanning tree in the weighed networks using the obtained eigenvalues, which completely agrees with the result deuced in a direct way, indicating that our computation for the eigenvalues and their multiplicities is correct.

It should be mentioned that although we have only studied a particular family of weighted networks, by using the generalized corona product of weighted graphs~\cite{LaJaKi16} and the edge weight reinforcement mechanism~\cite{BaBaVe04,BaBaVe04PRE}, one can easily generate various weighted complex networks, with their features qualitatively similar to those of the weighted model considered here. Since our model is exactly solvable, it provides a good facility to study
analytically various dynamical processes taking place upon it,  unveiling the effects of heterogenous weight distribution on these processes.

\ack{This work is supported by the National Natural Science Foundation of China under Grant No. 11275049.}


\end{document}